\documentclass[conference]{IEEEtran}
\IEEEoverridecommandlockouts
\usepackage{cite}
\usepackage{amsmath,amssymb,amsfonts}
\usepackage{amsthm}
\usepackage{graphicx}
\usepackage{subcaption}  
\usepackage{textcomp}
\usepackage{xcolor}
\usepackage{nicefrac}
\usepackage[bookmarks=false]{hyperref}
\usepackage{enumitem}
\usepackage{tikz}
\usetikzlibrary{shadows}
\usepackage{algorithm}
\usepackage[noend]{algpseudocode}
\hypersetup{draft}
\algrenewcommand\alglinenumber[1]{\scriptsize#1:}
\algnewcommand{\LineComment}[1]{\Statex \(\triangleright\) #1}
\usetikzlibrary{shapes.geometric, arrows.meta, positioning, fit}
\newtheorem{proposition}{Proposition}
\def\BibTeX{{\rm B\kern-.05em{\sc i\kern-.025em b}\kern-.08em
    T\kern-.1667em\lower.7ex\hbox{E}\kern-.125emX}}
\begin{document}

\title{RIDAS: A Multi-Agent Framework for AI-RAN with Representation- and Intention-Driven Agents\\
}

\author{
Kuiyuan Ding\textsuperscript{*}, Caili Guo\textsuperscript{†}, Yang Yang\textsuperscript{*} and Jianzhang Guo\textsuperscript{‡}\\
\textsuperscript{*} Beijing Key Laboratory of Network System Architecture and Convergence, School of Information and \\ Communication Engineering, Beijing University of Posts and Telecommunications, Beijing 100876, China \\
\textsuperscript{†} Beijing Laboratory of Advanced Information Networks, School of Information and Communication Engineering, \\ Beijing University of Posts and Telecommunications, Beijing 100876, China \\
\textsuperscript{‡}\ ‌China Telecom Digital Technology Co., Ltd.‌, Beijing 100035, China \\
\{dingkuiyuan, guocaili, yangyang01\}@bupt.edu.cn, guojz6@chinatelecom.cn
}

\maketitle

\begin{abstract}
Sixth‐generation (6G) networks demand tight integration of artificial intelligence (AI) into radio access networks (RANs) to meet stringent quality‐of‐service (QoS) and resource‐efficiency requirements. Existing solutions struggle to bridge the gap between high‐level user intents and the low‐level, parameterized configurations required for optimal performance. To address this challenge, we propose RIDAS, a multi‐agent framework composed of representation‐driven agents (RDAs) and an intention‐driven agent (IDA). RDAs expose open interface with tunable control parameters—rank and quantization bits, enabling explicit trade‐offs between distortion and transmission rate. The IDA employs a two‐stage planning scheme (bandwidth pre‐allocation and reallocation) driven by a large language model (LLM) to map user intents and system state into optimal RDA configurations. Experiments demonstrate that RIDAS supports 36.47\% more users than WirelessAgent under equivalent QoS constraints. These results validate ability of RIDAS to capture user intent and allocate resources more efficiently in AI‐RAN environments. Code is available on: \href{https://github.com/echojayne/RIDAS.git}{https://github.com/echojayne/RIDAS.git}
\end{abstract}

\begin{IEEEkeywords}
6G, AI-RAN, AI agents, resource allocation, large language models.
\end{IEEEkeywords}

\section{Introduction}
Sixth‐generation (6G) networks envision a profound integration of artificial intelligence (AI) into communication infrastructures, thereby imposing more stringent requirements on the radio access network (RAN). In response, the concept of AI‐enabled RAN (AI‐RAN) has emerged, leveraging advanced AI methodologies to endow the RAN with enhanced intelligence—enabling higher resource utilization and improved quality of service (QoS) \cite{kundu2025airantransformingranaidriven}. AI‐RAN aspires to automate network management by translating high‐level business objectives or user intents into concrete network configurations and policy directives. However, the use of conventional AI techniques in AI‐RAN reveals a substantial gap between user intents, typically expressed in natural language, and the complex, parameterized configurations required for RAN deployment.

As AI technologies and hardware have advanced, LLMs have demonstrated exceptional performance in general-purpose domains. Their strong capability for intent understanding makes them promising candidates for enhancing AI-RAN. However, standalone LLMs face three key limitations: They cannot efficiently process multimodal data, dynamically decompose complex tasks, or interface with specialized tools, all of which impede their deployment in complex wireless environments \cite{tong2024wirelessagentlargelanguagemodel}.

To overcome these obstacles, LLM-based agents have emerged. By augmenting LLMs with modular functionalities—such as environmental data perception and external tool integration—these agents can interpret complex wireless contexts, make informed decisions, and execute appropriate actions. There has been many preliminary explorations that engage LLM-based agents in RAN. The authors in \cite{bao2025llmhricllmempoweredhierarchicalran} proposed the LLM-empowered hierarchical RAN intelligent controllers (RICs) (LLM-hRIC) framework to improve the collaboration between RICs in open RAN. In \cite{tong2024wirelessagentlargelanguagemodel}, the authors introduce a framework called WirelessAgent harnessing LLMs to create autonomous AI agents for diverse wireless network tasks.

Despite LLM-based agents advances, existing studies largely overlook how high-level, intention-driven LLM agents which serves as the “brain” of AI-RAN can effectively orchestrate underlying AI models, particularly those for data representation that lack natural-language capabilities. Specifically, there are two primary challenges that must be addressed:

\begin{itemize}
    \item Challenge 1: How can underlying AI models be designed with an open interface that exposes tunable operational parameters, enabling external, high-level control over the trade-off between resource efficiency and QoS?
    \item Challenge 2: Given such a controllable interface, what framework should an LLM-based agent employ to translate high-level user intents into a sequence of concrete control actions, dynamically adapting to network status to orchestrate the underlying models?
\end{itemize}

In this article, we propose RIDAS, a multi‐agent framework for AI‐RAN that consists of representation‐driven agents (RDAs) and an intention‐driven agent (IDA). RIDAS enables the RAN to allocate bandwidth resources efficiently so as to serve as many users as possible while satisfying their QoS (task performance) requirements. Specifically, we first design the RDA, deployed at the user end, which employs sign‐value‐independent decomposition (SVID) \cite{xu2024onebitextremelylowbitlarge} to represent the source; by adjusting the decomposition rank and the number of quantized bits, both the transmission rate and the QoS performance of the RDA can be precisely controlled. Second, we introduce the IDA, which takes the QoS requirements of users as the intention input and employs a two‐stage planning scheme—bandwidth pre‐allocation and bandwidth reallocation—to dynamically adjust the control parameters of RDAs, thereby minimizing bandwidth consumption while meeting user QoS demands. Experimental results show that RIDAS can support 36.47\% more users than the WirelessAgent framework under equivalent QoS constraints, demonstrating that RIDAS effectively captures user intent and allocates system resources more efficiently.

\section{System Model and Problem Formulation}\label{sec2}
\begin{figure}[t]
    \centering
    \includegraphics[width=\linewidth]{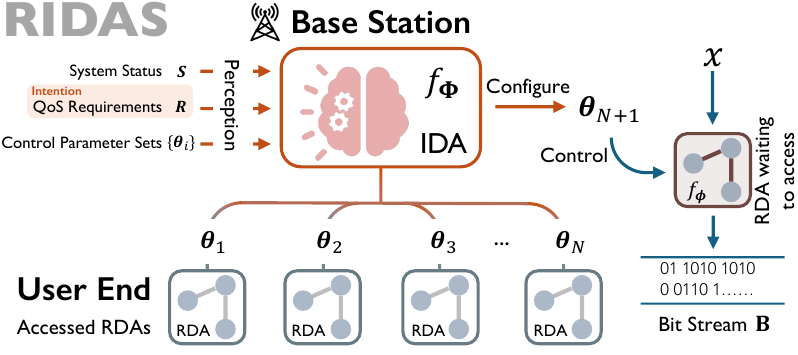}
    \caption{Overall architecture of proposed RIDAS framework.}
    \label{fig:system_model}
\end{figure}

The overall architecture of RIDAS framework is shown in Fig. \ref{fig:system_model}. At the user end (UE), there are $N$ RDAs connected to the base station (BS), each governed by a distinct control parameter $\boldsymbol{\theta}_1, \boldsymbol{\theta}_2, \cdots, \boldsymbol{\theta}_N$. When a new RDA attempts to access the BS, the IDA, deployed at the BS, determines its control parameter $\boldsymbol{\theta}_{N+1}$ based on the current system state $S$, the parameter sets of the existing RDAs $\left\{\boldsymbol{\theta}_i \mid i = 1, \cdots, N\right\}$ and the specific requirements $R$ of the new RDA (e.g., QoS constraints). Denoting the IDA as a mapping function $f_{\boldsymbol{\Phi}}$ parameterized by $\boldsymbol{\Phi}$, this process can be formulated as follows:
\begin{equation}
\boldsymbol{\theta}_{N+1} = f_{\boldsymbol{\Phi}} \big(S, R, \left\{\boldsymbol{\theta}_i\right\} \big).
\end{equation}

The configured control parameter $\boldsymbol{\theta}_{N+1}$ governs how the new RDA represents messages. Specifically, when a message $x$ needs to be encoded, the RDA denoted by $f_{\boldsymbol{\phi}}$ and parameterized by $\boldsymbol{\phi}$ generates the encoded bit stream $\mathcal{B}$ under the guidance of $\boldsymbol{\theta}_{N+1}$. This process is formulated as:
\begin{equation}
\mathcal{B} = f_{\boldsymbol{\phi}}(x;\boldsymbol{\theta}),
\end{equation}
where $\mathcal{B}$ represents the output bit stream.

The design of the RDA should ensure that both the quality and the quantity of the generated representations can be effectively controlled by the associated control parameter $\boldsymbol{\theta}_{N+1}$. Here, the quality of the representations reflects the level of distortion, while the quantity measured by the length of the encoded bit stream $\mathcal{B}$ corresponds to the transmission rate. Therefore, Challenge 1 lies in identifying a feasible mapping that jointly satisfies these objectives by balancing representation fidelity and communication efficiency.

The objective of designing the IDA is to determine an optimal control parameter $\boldsymbol{\theta}_{N+1}$ that achieves a balance between transmission rate and distortion. This goal can be formulated as the following constrained optimization problem:

\begin{equation}
\label{optimization_problem}
\begin{aligned}
\boldsymbol{\theta}_{N+1}^*
&= \arg\min_{\boldsymbol{f_{\boldsymbol{\Phi}}}}\;
\underbrace{\mathbb{E}_{x\sim p_X}[\,\lvert \mathcal{B}\rvert\,]}_{\mathcal{R}(\boldsymbol{\theta}_{N+1})} \\
\text{s.t.}\quad
&\sum_{i=1}^{N+1}\mathcal{R}(\boldsymbol{\theta}_{i}) \;\le\; B_{\max}(S),
\\
&\;\;\mathcal{B} = f_{\boldsymbol{\phi}}(x;\boldsymbol{\theta}_{N+1}),\\
&\;\;D(\boldsymbol{\theta}_{N+1}) \;\le\; D_{\mathrm{req}}(R),\\
&\;\;\boldsymbol{\theta}_{N+1} = f_{\boldsymbol{\Phi}}\big(S, R, \left\{\boldsymbol{\theta}_i,i=1,\cdots,N\right\}\big),
\end{aligned}
\end{equation}
where $D(\boldsymbol{\theta})$ and $\mathcal{R}(\boldsymbol{\theta})$ denote the average distortion and rate under the control parameter $\boldsymbol{\theta}$, respectively, $B_{\max}(S)$ is the total bandwidth budget imposed by the current system state $S$, $D_{\mathrm{req}}(R)$ is the distortion requirement determined by the QoS constraint $R$ of the new RDA and $\lvert \mathcal{B} \rvert$ denotes the length of the generated bit stream $\mathcal{B}$.

By solving the optimization problem in Eq.~\eqref{optimization_problem} through the design of a suitable IDA function $f_{\boldsymbol{\Phi}}$, we aim to effectively address Challenge 2.

\section{Proposed RIDAS Framework}
\begin{figure*}[t]
    \centering
    \includegraphics[width=0.8\linewidth]{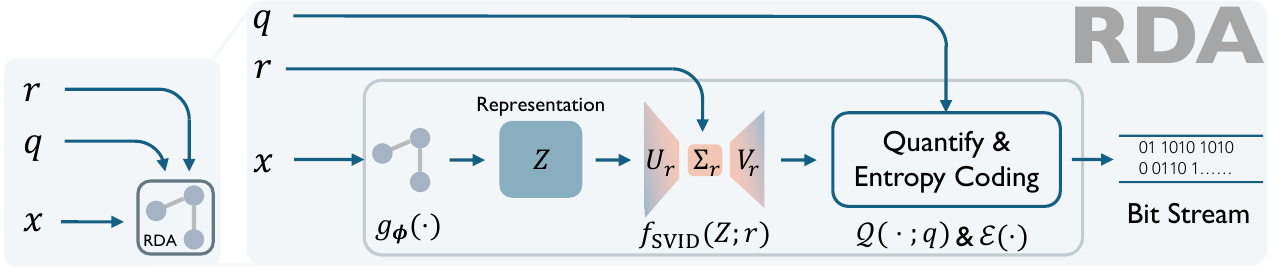}
    \caption{The architecture of RDA.}
    \label{fig:RDA}
\end{figure*}
The overall architecture of the proposed RIDAS framework is illustrated in Fig.~\ref{fig:system_model} which comprises two primary components: RDAs at the user end and an IDA at BS. The RDAs expose interfaces parameterized by $\boldsymbol{\theta}$, which the IDA configures by computing a near-optimal control parameter $\boldsymbol{\theta}^*$. This process ensures that user QoS requirements are satisfied while minimizing system resource utilization.

In this section, we first present the design of the RDAs, and then provide a detailed description of the overall IDA architecture.
\subsection{The design of RDA}
As illustrated in Fig. \ref{fig:RDA}, the RDA architecture comprises a well-trained deep neural network (DNN) denoted as $g_{\boldsymbol{\phi}}$ parameterized by $\boldsymbol{\phi}$, followed by the SVID module and subsequent quantization and entropy‐coding stages.
\subsubsection{DNN}
The DNN within the RDA serves as the core component for data representation and fundamentally determines its effectiveness. In our framework, we abstract away the specifics of the DNN architecture and training procedure, assuming that it produces sufficiently efficient representations. Instead, we concentrate on the downstream processing of these representations, providing a programmable interface that enables precise control over both representation quality and transmission efficiency.
\subsubsection{SVID}
The representation space produced by the DNN is typically very high-dimensional, and directly transmitting these raw representations would incur a substantial bandwidth overhead. To mitigate this, we innovatively leverage the SVID method introduced in \cite{xu2024onebitextremelylowbitlarge}, which was originally developed for decomposing DNN weights. Concretely, let the output representation of the DNN be denoted as $Z \in \mathbb{R}^{m \times n}$, and choose a target rank $r \ll \min(m,n)$. We then form the following rank-r approximation:

\begin{equation}
\label{SVID}
    Z \approx Z_{\rm sign} \odot \big(U_r\Sigma_r V_r^T \big),
\end{equation}

In Eq. (\ref{SVID}), $Z_{\rm sign} = \mathrm{sign}(Z)\in\{-1,1\}^{m\times n}$ is the elementwise sign matrix of $Z$, where
$$\mathrm{sign}(z_{ij}) =
\begin{cases}
+1, & z_{ij}\ge0,\\
-1, & z_{ij}<0.
\end{cases}$$
The operator $\odot$ denotes the Hadamard (elementwise) product.  The factorization $U_r\Sigma_r V_r^{T}$ is the first $r$ singular value decomposition (SVD) of $\lvert Z\rvert$, the matrix obtained by taking the absolute value of each entry of $Z$. The matrices $U_r\in\mathbb{R}^{m\times r}$ and $V_r\in\mathbb{R}^{r\times n} $ are the matrices of the first $r$ left and right singular vectors of $\lvert Z\rvert$, respectively.  Finally,
$\Sigma_r = \mathrm{diag}(\sigma_1,\sigma_2,\ldots,\sigma_{r})\in\mathbb{R}^{r\times r}$
is the diagonal matrix of singular values ordered as $\sigma_1\ge\sigma_2\ge\cdots\ge\sigma_r\ge0$.

The following proposition is the reason why we choose SVID over the conventional SVD method.

\begin{proposition}\label{prop:2}
Given a matrix $\mathbf{W}$ and its element‐wise absolute value $|\mathbf{W}|$, let
\[
\mathbf{W}
=
\mathbf{W}_{\mathrm{sign}}\odot|\mathbf{W}|.
\]
We decompose these as
\[
\mathbf{W}
=
\mathbf{a}\mathbf{b}^\top + \mathbf{E}_1,
\quad
|\mathbf{W}|
=
\tilde{\mathbf{a}}\tilde{\mathbf{b}}^\top + \mathbf{E}_2,
\]
where $a,b$ and $\tilde a, \tilde b$ are vectors of appropriate dimensions given by SVD, each $\mathbf{E}_i$ is the corresponding error matrix.  In terms of the Frobenius norm, the SVID approximation is closer to the original matrix $\mathbf{W}$:
\begin{equation}\label{eq:prop2}
\bigl\lVert
\mathbf{W}
-
\mathbf{W}_{\mathrm{sign}}\odot\tilde{\mathbf{a}}\tilde{\mathbf{b}}^\top
\bigr\rVert_{F}^{2}
\le
\bigl\lVert
\mathbf{W}
-
\mathbf{a}\mathbf{b}^\top
\bigr\rVert_{F}^{2}.
\end{equation}
\end{proposition}
\begin{proof}
For a detailed proof of this proposition, we refer the reader to the work of Xu et al. \cite{xu2024onebitextremelylowbitlarge}.
\end{proof}
As shown in Eq.(\ref{eq:prop2}), the reconstruction error of SVID is provably no greater than that of SVD. Furthermore, Proposition \ref{prop:2} can be straightforwardly generalized to rank-$r$ decompositions. Accordingly, SVID not only achieves a lower-error approximation $U_r\Sigma_r V_r^{T}$ than SVD, but also offers a controllable trade-off between representation dimension and distortion, thereby providing a preliminary solution to Challenge 1. The computational complexity of SVID is dominated by the truncated SVD step, resulting in an overall complexity of $\mathcal{O}(mnr)$. 

\subsubsection{Quantization and Entropy‐Coding}

For digital transmission, we further apply quantization and entropy encoding to the matrices $U_r$, $\Sigma_r$, and $V_r^{T}$. Specifically, we denote the SVID-based low-rank approximation as $f_{\rm SVID}(Z; r)$, the quantization process as $\mathcal{Q}(W; q)$ where $q$ indicates the number of quantization bits, and the subsequent entropy encoding as $\mathcal{E}(W_{\rm Q})$, where $W_{\rm Q}$ denotes the quantized representation. Accordingly, the RDA parameterized by $\boldsymbol{\phi}$ can be expressed as:
\begin{equation}
\label{rda}
\begin{aligned}
    \mathcal{B} & = f_{\boldsymbol{\phi}}(x;\boldsymbol{\theta}) \\
               & = \mathcal{E}\big( \mathcal{Q}\big( f_{\rm SVID}\big(g_{\boldsymbol{\phi}}(x);r\big);q \big) \big),
\end{aligned}
\end{equation}
where the control parameter $\boldsymbol{\theta}=\left\{ r,q \right\}$.

Eq. (\ref{rda}) presents a feasible implementation of the RDA. The distortion and transmission rate of the resulting representation bit stream can be effectively controlled by rank $r$ and quantization bits $q$. Specifically, larger values of $r$ and $q$ lead to higher transmission rates but lower task performance distortion, while smaller values of $r$ and $q$ result in reduced transmission rates at the cost of increased task performance distortion. Thus, this implementation provides a controllable trade-off between  resource efficiency and QoS. Up to this point, we have presented a viable instantiation of the RDA, thereby addressing Challenge 1.

\subsection{The design of IDA}

\begin{figure*}[t]
    \centering
    \includegraphics[width=0.9\linewidth]{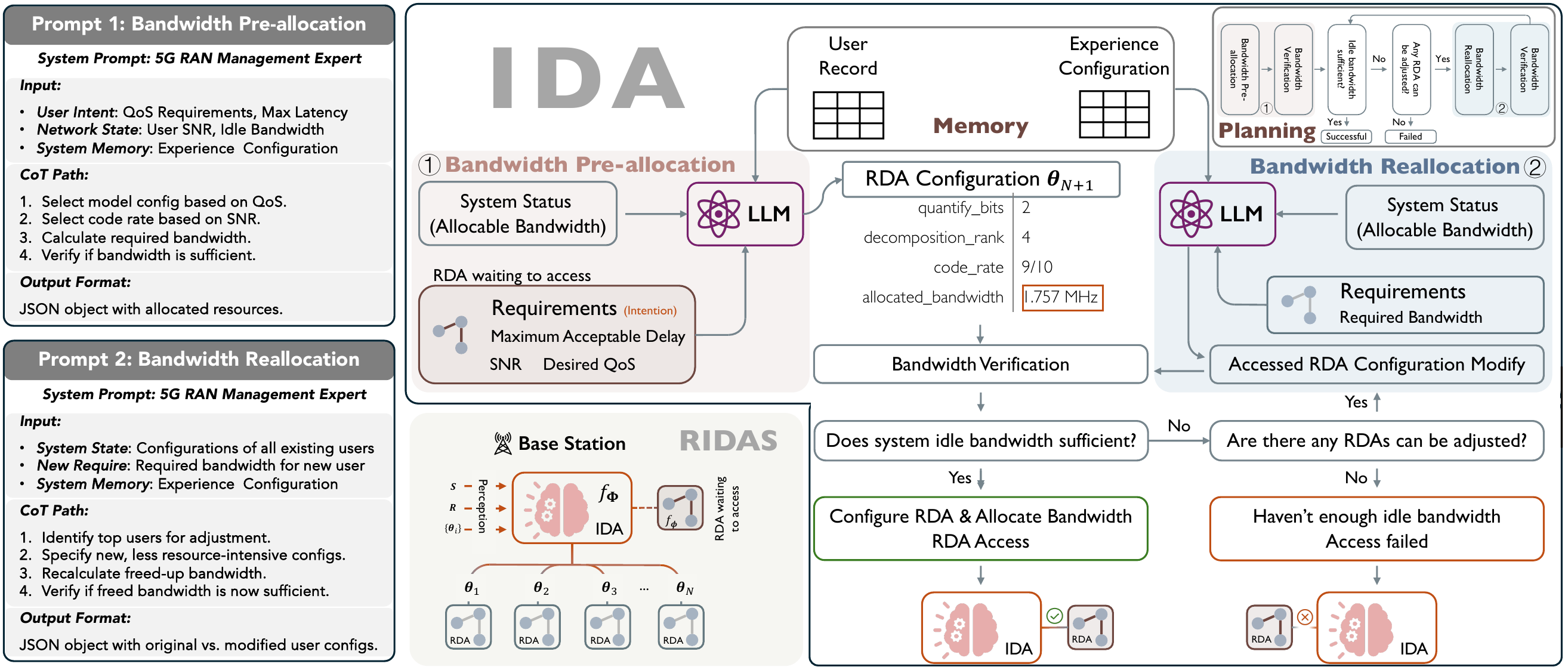}
    \caption{Overall architecture of proposed IDA.}
    \label{fig:IDA}
    \vspace{-10pt}
\end{figure*}

As illustrated in Eq. (\ref{optimization_problem}), IDA seeks to determine the optimal control parameter $\boldsymbol{\theta}_{N+1}^*$ for the new RDA, such that the resulting representation bit stream length $\lvert\mathcal{B}\rvert$ is minimized while satisfying the maximum‐distortion constraint of the RDA. As illustrate in Fig. \ref{fig:IDA}, we realize IDA as an LLM‐based agent composed of two principal routines: (1) bandwidth pre‐allocation and (2) bandwidth reallocation. Specifically, the IDA first selects $\boldsymbol{\theta}_{N+1}^*$ for the new RDA and then allocates the corresponding transmission bandwidth. In the following subsections, we first present the memory module of IDA, and then describe its planning pathway in detail.

\subsubsection{Memory}
The memory module of IDA persistently stores the current configuration set, i.e., the control parameters of all connected RDAs ${\boldsymbol{\theta}_i}$, together with the allocated bandwidth of each RDA (\textit{User Record} table in Fig. \ref{fig:IDA}) and the corresponding achieved experimental bit stream length and distortion under those configurations (\textit{Experience Configuration} table in Fig. \ref{fig:IDA}). All of this information is maintained in a real-time updating database, which serves as the perception input to the LLM-based agent.

\subsubsection{Planning}
In order to enable the LLM to generate near-optimal configurations, we have designed a dedicated planning pathway for IDA.

When a new RDA requests access to the BS, the LLM within IDA first proposes control parameter $\boldsymbol{\theta}_{N+1}$ and pre-allocates bandwidth according to the QoS requirements of the RDA (including distortion and rate requirements) and the currently available idle bandwidth of the system. During this pre-allocation stage, the prompt steers the LLM to retrieve from past experience a configuration that minimizes transmission rate while satisfying the distortion constraint. Owing to its strong instruction-following and contextual-understanding capabilities, the initial proposal of the LLM typically lies very close to the true optimum, i.e., $\boldsymbol{\theta}_{N+1}\approx\boldsymbol{\theta}_{N+1}^*$, and thus the resulting bandwidth allocation effectively balances system resource usage against user demand.

However, since LLM are prone to numerical hallucinations when performing precise calculations, we further validate the pre-allocated bandwidth by measuring the empirical transmission rate under the proposed configuration before committing to the final assignment.

If the idle bandwidth of system is insufficient to satisfy the pre-allocation request, IDA initiates a reallocation procedure across all RDAs connected to the BS rather than rejecting the new RDA directly. Because the control parameters and bandwidth assignments determined during the pre-allocation stage may not be optimal, some RDAs may hold redundant capacity. Consequently, when idle bandwidth is inadequate, the LLM of IDA evaluates whether the configuration of any connected RDA can be adjusted to free up additional resources. If such opportunities exist, IDA modifies those configurations and reallocates their bandwidth. It then reassesses the available bandwidth. If the newly available capacity meets the requirements, the new RDA is configured, and the corresponding bandwidth is provisioned for it. Otherwise, the reallocation stage repeats until either sufficient idle bandwidth is secured (connected successfully) or no further adjustments are possible (connected failed).

By employing the two-stage planning pathway of bandwidth pre‐allocation and reallocation, IDA is able to assign near‐optimal control parameters and bandwidth resources so as to minimize the representation bit‐stream length $\lvert\mathcal{B}\rvert.$ This not only conserves system resources but also ensures that the QoS requirements of each RDA are met, thereby directly addressing the optimization objective in Eq. (\ref{optimization_problem}). The concrete instantiation of IDA presented here thus constitutes a viable solution to Challenge 2.

The proposed RDA and IDA together form the RIDAS framework. On the one hand, RDAs expose an interface for controlling their representation actions, thereby satisfying diverse user QoS requirements with minimal resource consumption by adjusting the control parameter $\boldsymbol{\theta}$. On the other hand, the IDA interprets user intent to configure the control parameters of RDAs in a manner that further minimizes system resource usage. Consequently, RIDAS introduces a novel AI-RAN paradigm, demonstrating how LLM-based agents can interact with and manage the underlying AI models.

\section{Experimental Results}
\subsection{Settings}
\subsubsection{Scenario setup}
In our experiments, the total available bandwidth is set to 100 MHz. The signal-to-noise ratio (SNR) for each RDA is randomly generated in the range of 5 dB to 30 dB. The RDA code rate is selected from the set
$\left\{\tfrac12,\;\tfrac35,\;\tfrac23,\;\tfrac34,\;\tfrac45,\;\tfrac56,\;\tfrac89,\;\tfrac{9}{10}\right\}$.
According to \cite{9389782}, the end-to-end delay in 6G networks is expected to be less than 1 ms; accordingly, we impose a maximum allowable transmission delay per RDA in the interval [0.05 ms, 0.5 ms]. Moreover, we characterize  QoS requirement of each RDA by its top-1 classification accuracy:
\begin{itemize}
    \item Low: QoS demands accuracy $>$ 70\%, corresponding to shorter bit-stream lengths;
    \item Medium: QoS requires accuracy $>$ 80\%;
    \item High: QoS requires accuracy $>$ 90\%, corresponding to longer bit-stream lengths.
\end{itemize}

Under these configurations, the required bandwidth is computed as follows:
\begin{small}
\begin{equation}
\text{Bandwidth (MHz)}
= \frac{\lvert \mathcal{B} \rvert \,/\, \text{code rate}}{\text{transmission time}}
\times \frac{1}{\log_2\bigl(1 + 10^{\frac{\text{SNR}}{10}}\bigr)\times 10^6}\,.
\end{equation}
\end{small}

\subsubsection{Task and model setup}
To evaluate the effectiveness of our RDA design, we perform an image-classification task on the CIFAR-10 dataset \cite{krizhevsky2009learning}, using a ViT-B/16 architecture from the CLIP framework \cite{radford2021learningtransferablevisualmodels} as the representation backbone $g_{\boldsymbol{\phi}}$. The LLM served in IDA is set to DeepSeek-V3-0324 \cite{deepseekai2025deepseekv3technicalreport}.
\begin{figure}[t]
    \centering
    \includegraphics[width=0.8\linewidth]{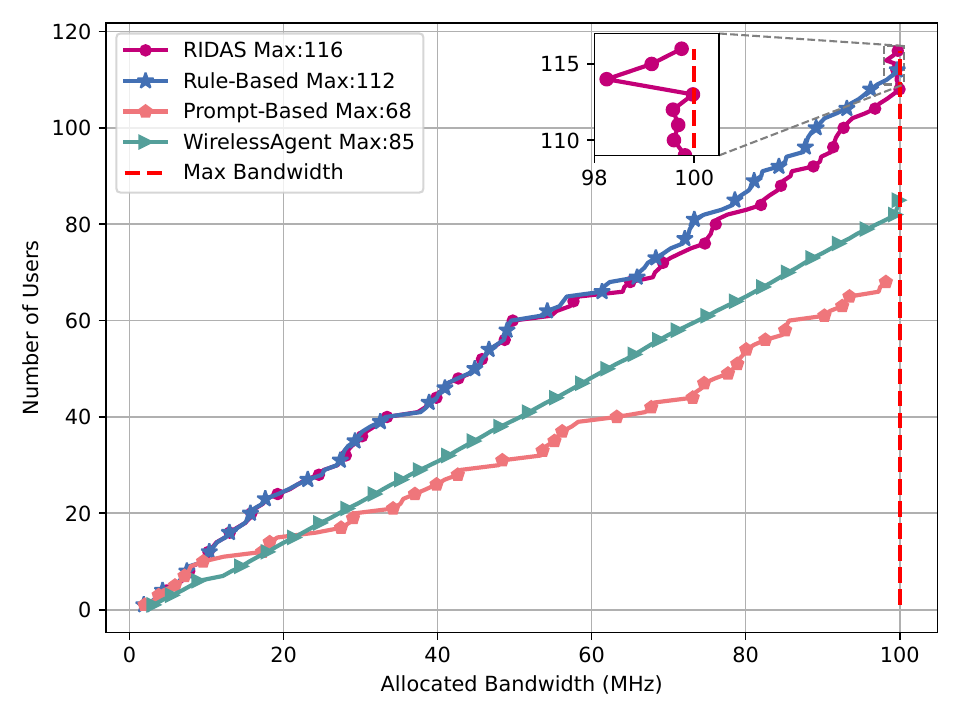}
    \caption{Number of users connected to BS of different methods with the same bandwidth.}
    \label{fig:bd_vs_use_number}
\end{figure}

\subsection{Baselines}
For fair comparisons, we employ the following baselines:
\begin{itemize}
    \item WirelessAgent \cite{tong2024wirelessagentlargelanguagemodel}: An LLM-based autonomous agent for wireless tasks.
    \item Prompt-Based: A simplified version of our method that uses a single LLM prompt for allocation, omitting the verification and reallocation stages.
    \item Rule-Based: A heuristic approach that allocates bandwidth based on optimal control parameters and an SNR-scaled code rate.
\end{itemize}

\subsection{Results}


\begin{figure}[t]
    \centering
    \includegraphics[width=0.8\linewidth]{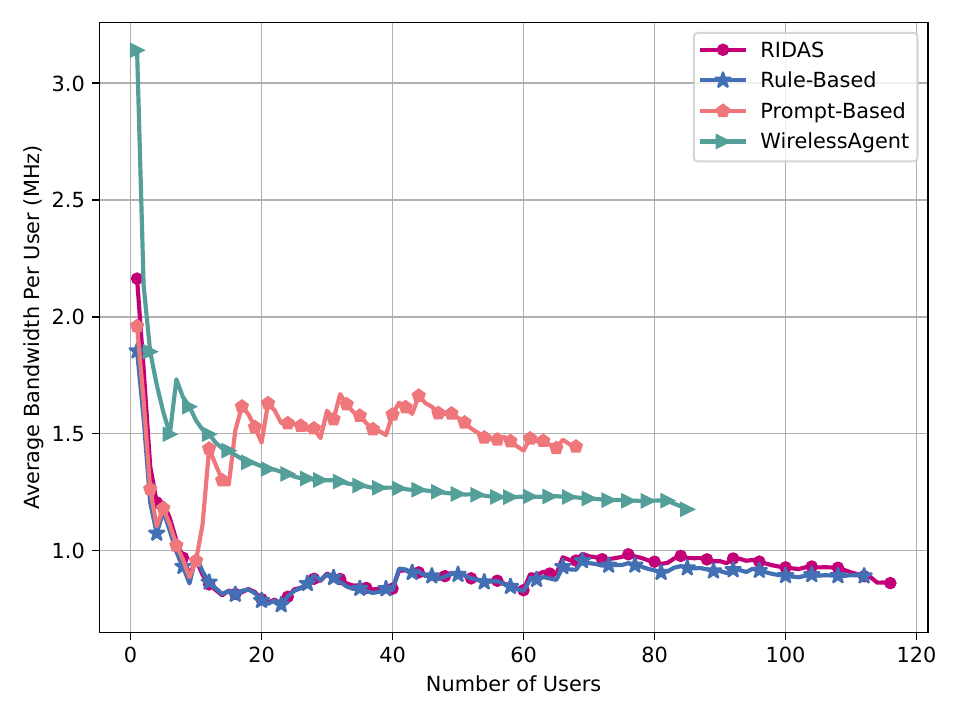}
    \caption{Average bandwidth of connected users.}
    \label{fig:ave_bandwidth}
    \vspace{-15pt}
\end{figure}

We simulate a common queue of users awaiting connection across all methods. Hence, under a fixed total bandwidth, supporting a greater number of users or, for the same number of connected users, allocating a lower average bandwidth to each user while still meeting their QoS requirements indicates more efficient resource utilization.

Fig. \ref{fig:bd_vs_use_number} shows the number of users that can be supported by the BS under different baseline methods, given an identical total bandwidth budget. As illustrated in Fig. \ref{fig:bd_vs_use_number}, when the allocated bandwidth approaches 100\% of the system capacity, the proposed RIDAS accommodates up to 116 users, compared with 112 for the rule‐based baseline, 85 for the WirelessAgent framework, and 68 for the prompt‐based scheme. 

\begin{figure*}[t]
    \centering
    \includegraphics[width=0.75\linewidth]{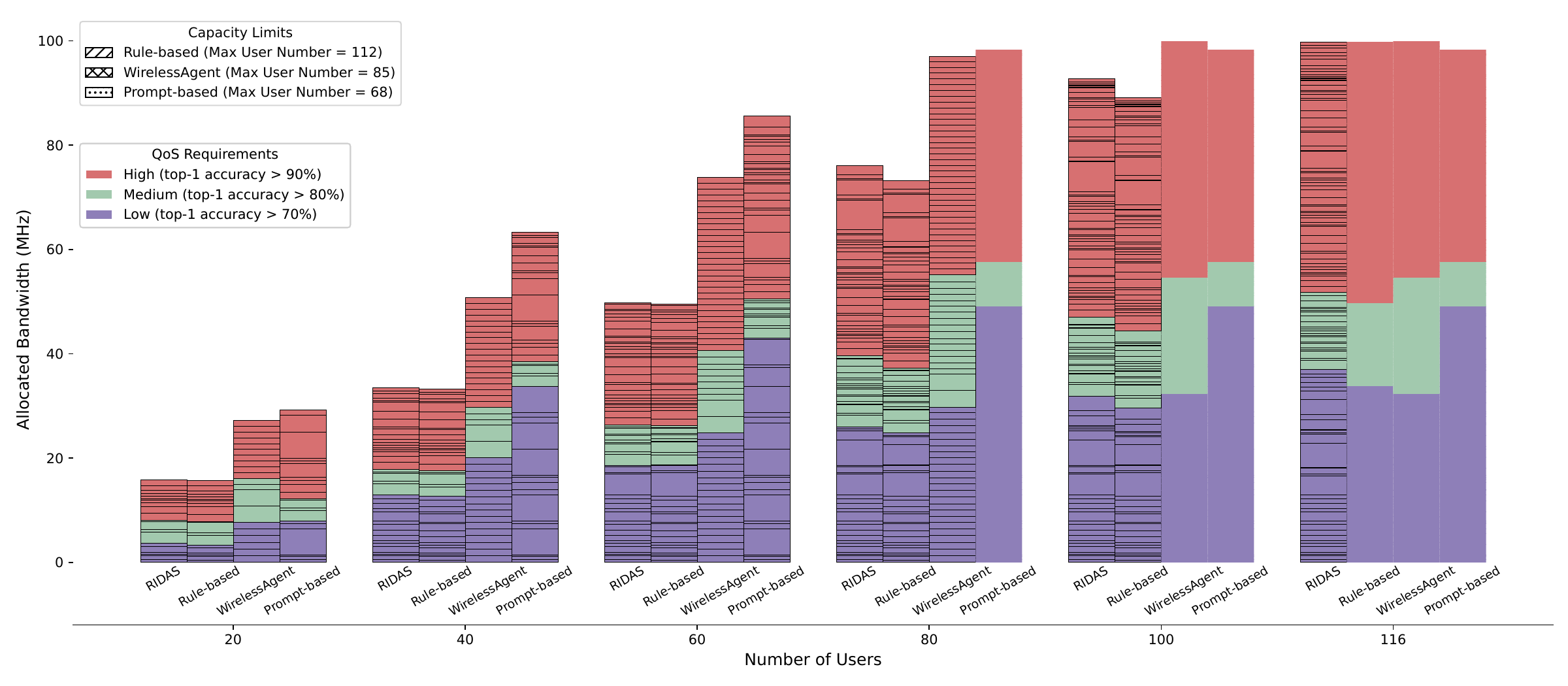}
    \caption{Details of bandwidth allocation.}
    \label{fig:details_of_bandwidth_allocation}
\end{figure*}

\begin{figure}[t]
    \centering
    \includegraphics[width=0.8\linewidth]{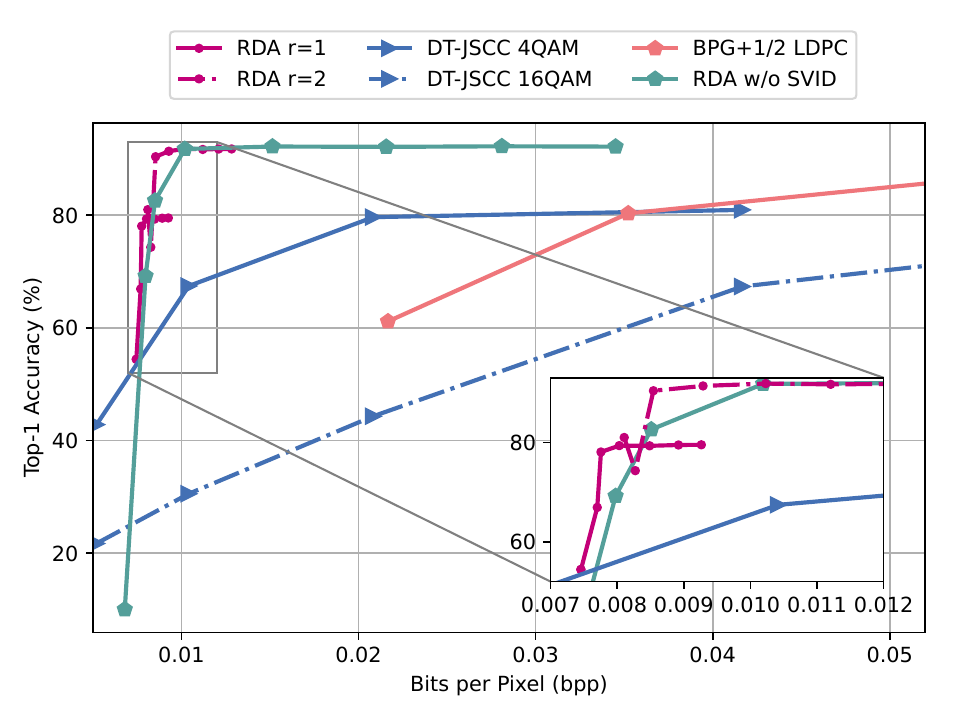}
    \caption{Accuracy versus bpp.}
    \label{fig:acc_bpp}
    \vspace{-5pt}
\end{figure}

These results indicate that the proposed IDA not only attains near-optimal control-parameter configurations but also supports a greater number of concurrent users than the rule-based scheme. This improvement stems from the fact that the rule-based method selects code rates by linearly scaling with SNR, whereas IDA adaptively determines code rates based on empirical performance data—thereby conserving bandwidth more effectively across diverse scenarios.

Furthermore, for RIDAS, when 109 users are connected, the allocated bandwidth is nearly exhausted. Upon the 110th connection attempt, IDA triggers its reallocation stage to reclaim additional bandwidth from existing RDAs, thereby freeing sufficient capacity to admit the additional user.

Fig. \ref{fig:ave_bandwidth} shows the average bandwidth allocated per user as the connected users varies. At 85 concurrent users, the mean per-user allocation is 0.968 MHz under RIDAS, compared with 0.925 MHz for Rule-Based method and and 1.175 MHz for WirelessAgent framework. These results further demonstrate that the proposed IDA distributes bandwidth more efficiently while still meeting QoS requirements of each user. 

Fig. \ref{fig:details_of_bandwidth_allocation} details the per-user bandwidth assignments for varying acceptable transmission-latency requirements across different number of connected user. As shown in Fig. \ref{fig:details_of_bandwidth_allocation} , RIDAS dynamically adapts allocation according of each RDA to its latency tolerance—users that can tolerate higher delays are provisioned with less bandwidth, whereas those requiring lower latency receive larger allocations—thereby conserving resources. By contrast, the WirelessAgent framework tends to distribute bandwidth uniformly among all users. These findings further demonstrate the ability of IDA within RIDAS to infer the intent of each agent and generate control‐parameter configurations that optimize resource utilization.

To demonstrate the effectiveness of the proposed RDA, we compare its performance against two key baselines. Fig.~\ref{fig:acc_bpp} illustrates this comparison by plotting the classification accuracy as a function of bits per pixel (bpp)\footnote{Here, $\mathrm{bpp} = \frac{\text{total transmission bits}}{3 \times \text{Height} \times \text{Width}}$, where Height and Width denote the image dimensions.}. The baselines used for this comparison are DT-JSCC~\cite{10159007}, a task-oriented communication scheme employing digital modulation, and BPG + 1/2 LDPC, which represents a conventional approach of transmitting images using BPG coding protected by a rate-1/2 LDPC code. As shown in Fig.~\ref{fig:acc_bpp}, at approximately 0.008 bpp, RDA outperforms DT-JSCC with 4-QAM and the RDA variant without SVID by over 20\% and 18\% in accuracy, respectively. These results demonstrate that RDA effectively preserves the quality of the representation—and hence the downstream task performance—even in the low-bpp regime. At higher bpp values, where representational distortion becomes negligible, accuracy of RDA approaches that of the original representation model.

\section{Conclusion}


In this work, we have introduced RIDAS, a multi-agent framework for AI-RAN that unifies low-level representation control with high-level intent interpretation via its RDA and IDA components. In our evaluation, RIDAS dynamically adjusted its control parameters in response to network conditions and user QoS requirements, thereby maximizing resource utilization. Its two-stage planning process, comprising bandwidth pre-allocation and subsequent reallocation, achieved near-optimal performance by satisfying both transmission-rate and task-performance demands. As a contribution, RIDAS offers a novel paradigm for autonomous, intent-driven RAN management and provides a promising foundation for 6G networks. Future work will extend the framework to incorporate end-to-end delay control and adapt to more complex deployment scenarios, further enhancing its practicality in real-world wireless environments.

\bibliographystyle{IEEEtran}
\bibliography{reference}
\end{document}